\pdfoutput=1
\documentclass[12pt, oneside]{amsart}

\usepackage{hyperref}
\hypersetup{%
	pdfpagemode={UseOutlines},
	bookmarksopen,
	pdfstartview={FitH},
	colorlinks,
	linkcolor={blue},
	citecolor={blue},
	urlcolor={blue}
}

\usepackage{geometry} 

\usepackage{dcolumn}
\newcolumntype{d}{D{.}{.}{-1} } 

\usepackage{amsmath, amssymb, amsthm}
\usepackage{mathrsfs}
\usepackage{enumitem}
\usepackage{aligned-overset}
\usepackage{mathtools}
\usepackage[numbers]{natbib}

\numberwithin{equation}{section}

\newcommand\eps\varepsilon

 \newcommand{\A}{\mathcal{A}}
 
\newtheorem{theorem}{Theorem}[section]
\newtheorem{definition}[theorem]{Definition}
\newtheorem{lemma}[theorem]{Lemma}
\newtheorem{corollary}[theorem]{Corollary}
\newtheorem{proposition}[theorem]{Proposition}

\theoremstyle{remark}
\newtheorem{remark}[theorem]{Remark}

\newcommand{\R}{\mathbb{R}} 
\newcommand{\Q}{\mathbb{Q}} 
\newcommand{\Z}{\mathbb{Z}} 
\newcommand{\C}{\mathbb{C}} 

\title[The Hudson theorem in LCA groups and infinite quantum spin systems]{The Hudson theorem in LCA groups and infinite  quantum spin systems}
\author{Fabio Nicola}
\address[Fabio Nicola]{Dipartimento di Scienze Matematiche, Politecnico di Torino, Corso Duca degli Abruzzi 24, 10129 Torino, Italy}
\email{fabio.nicola@polito.it}

\author{Federico Riccardi}
\address[Federico Riccardi]{Dipartimento di Scienze Matematiche, Politecnico di Torino, Corso Duca degli Abruzzi 24, 10129 Torino, Italy}
\email{federico.riccardi@polito.it}

\begin{document}

    
    \begin{abstract}
    The celebrated Hudson theorem states that the Gaussian functions in $\R^d$ are the only functions whose Wigner distribution is everywhere positive. Motivated by quantum information theory, D. Gross proved an analogous result on the Abelian group $\mathbb{Z}_d^n$, for $d$ odd --- corresponding to a system of $n$ qudits --- showing that the Wigner distribution is nonnegative only for the so-called stabilizer states. Extending this result to the thermodynamic limit of finite-dimensional systems naturally leads us to consider general $2$-regular LCA groups that possess a compact open subgroup,  where the issue of the positivity of the Wigner distribution is currently an open problem. We provide a complete solution to this question by showing that if the map $x\mapsto 2x$ is measure-preserving, the functions whose Wigner distribution is nonnegative are exactly the subcharacters of second degree, up to translation and multiplication by a constant. Instead, if the above map is not measure-preserving, the Wigner distribution always takes negative values. We discuss in detail the particular case of infinite sums of discrete groups and infinite products of compact groups, which correspond precisely to infinite quantum spin systems. Further examples include $n$-adic systems, where $n\geq 2$ is an arbitrary integer (not necessarily a prime), as well as solenoid groups.
    
    \end{abstract}
    
    \maketitle
    
\section{Introduction}
Let $f\in L^2(\R^d)$. The \textit{Wigner distribution} of $f$ is the function $Wf$ in $\R^d\times\R^d$ given by 
\begin{equation}\label{eq wigner_rd}
Wf(x,\xi)=\int_{\R^d}f\Big(x+\frac{y}{2}\Big)\overline{f\Big(x-\frac{y}{2}\Big)} e^{-2\pi i \xi\cdot y}\, dy\qquad (x,\xi)\in \R^d\times\R^d. 
\end{equation}
Observe that $Wf$ is continuous and takes real values.

The concept of the Wigner distribution was introduced by E. Wigner in \cite{wigner}, intending to link a phase-space probability distribution to a quantum system (of which the function $f$ is the wave function in the Schr\"odinger representation). As might be expected, this perspective turned out to be enlightening in various ways; for instance, as noted in \cite{wigner}, the Wigner distribution of a free particle is simply transported along the corresponding Hamiltonian classical flow, in perfect agreement with the correspondence principle (see also \cite{cordero_rodino_CMP} for an elaboration of this idea in the presence of a potential). Since then, this phase-space picture has seen spectacular developments, with applications to mathematical physics, partial differential equations, harmonic analysis and mathematical signal processing; see \cite{abreu_2019,fefferman,grochenig_book,lerner_book,lieb_book,luef_2018,werner_1984} and especially \cite{degosson_wigner}.

One can see that if $\|f\|_{L^2}=1$ \textit{and} $Wf\in L^1(\R^d)$, then 
\[
\int_{\R^d\times\R^d} Wf(x,\xi)\, dx\,d\xi=1,
\]
as one expects from a probability distribution. However, it turns out that $Wf$ in general also takes negative values. Although this does not undermine the original idea (because only the averages of $Wf$ on the Planck scale are physically meaningful), it suggests the problem of characterizing the functions $f\in L^2(\R^d)$ so that
$Wf>0$ everywhere. The solution to this problem is provided by the celebrated Hudson theorem, which states that if $f\in L^2(\R^d)$, then $Wf>0$ everywhere in $\R^d\times\R^d$ if and only if $f$ has the form 
\[
f(x)= e^{-\pi x\cdot Ax+2\pi b\cdot x+c}
\]
where $A\in GL(d,\C)$ is an invertible $d\times d$ matrix over $\C$ with positive definite real part, $b\in \C^d$, $c\in\C$; see the original paper \cite{hudson} or \cite[Theorem 4.4.1]{grochenig_book}.

Motivated by quantum information theory, D. Gross \cite{gross2006hudson} later extended Hudson's theorem to systems of $n$ qudits, that is, to the finite Abelian group $\mathbb{Z}_d^n$, where $d\geq 3$ is odd. This restriction on $d$ is due to the fact that, as can be seen from the expression in \eqref{eq wigner_rd}, the definition of the Wigner distribution in $\R^d$ extends naturally to any locally compact Abelian (LCA) group, provided the map $x\mapsto 2x:=x+x$ is a topological isomorphism (see Definition \ref{def wigner}). When this condition is met, the group is called $2$-regular (see Definition \ref{def 2regular}) and $\mathbb{Z}_d^n$ is, indeed, $2$-regular if and only if $d$ is odd.
Now, it was proved in \cite{gross2006hudson} that, for $f\in \C^{dn}=\ell^2(\mathbb{Z}_d^n)$, we have $Wf\geq0$ everywhere if and only if $f$ is a so-called \textit{stabilizer state} --- a class of states first introduced by D. Gottesman in \cite{gottesman1998heisenberg} in the case of qubits ($d=2$), which have seen important applications in quantum error correction codes \cite{nielsen2010quantum}. It is remarkable how the (non)positivity of the Wigner distribution (in $\mathbb{Z}_d^n$) is closely linked to several other topics of quantum information and computation, such as classical simulation, quantum correlation, magic state
distillation and contextuality; see \cite{howard_2014,mari2012,siyouri_2016,veitch_2012}.

At this point, it is natural to ask what can be said, similarly, about the \textit{thermodynamic limit} of a system of $n$ qudits, as $n\to\infty$. The corresponding quantum systems, which have infinitely many degrees of freedom, are usually referred to as quantum spin systems \cite{bratteli} (see also \cite{naaijkens} for a gentle introduction). As we shall see, this problem naturally leads us to consider the framework of $2$-regular LCA groups that contain a compact open subgroup. In this generality, the issue of the positivity of the Wigner distribution is an open problem, although partial results are known --- e.g., the case of $2$-regular totally disconnected groups is considered in \cite{beny} assuming, however, that the function $f$ that satisfies $Wf\geq0$ is continuous. 
 In this note, we provide a complete solution to this problem. 

The characterization we are going to describe makes use of the notion of \textit{character of second degree} of an LCA group --- a concept first introduced by A. Weil \cite{weil64}, which will be recalled in Section \ref{sec caratteri} below. Also, in an LCA group, we say that a continuous function $\phi:A\to\C$ is a \textit{subcharacter of second degree} if there exists a compact open subgroup $H\subset A$ such that $\phi(x)=0$ for $x\in A\setminus H$ and the restriction of $\phi$ to $H$ is a character of second degree of $H$. We observe that subcharacters of second degree exist if and only if $A$ contains a compact open subgroup. We finally recall that, as already observed, in a 2-regular LCA group $A$, the Wigner distribution $Wf$ of a function $f\in L^2(A)$ is well defined as a function in $A\times\widehat{A}$, where $\widehat{A}$ denotes the dual group (see Definition \ref{def wigner} below). 
\begin{theorem}\label{thm mainteo 1} Let $A$ be a $2$-regular LCA group that contains a compact open subgroup. Suppose that the map $x\mapsto 2x$ is measure-preserving. 

Then, the following conditions are equivalent for $f\in L^2(A)$.
\begin{itemize}
    \item[(a)] $Wf(x,\xi)\geq0$ for every $(x,\xi)\in A\times\widehat{A}$.
\item[(b)] There exist  a subcharacter of second degree $\phi$ of $A$, $c\in\C$ and $y\in A$ such that
    \[
    f(x)=c\, \phi(x-y),\qquad \textrm{for almost every }x\in A. 
    \]
\end{itemize}
\end{theorem}

In the following, for the sake of brevity, functions $f$ as in Theorem \ref{thm mainteo 1} (b) (with $c\not=0$) will shortly be referred to as \textit{$S$-states} (see Definition \ref{def Sstate}). In fact, as a consequence of \cite[Theorem 1.4]{nicJMPA}, they coincide with the so-called stabilizer states, as defined in \cite{nicJMPA} in the general context of an LCA group (not necessarily $2$-regular) that contains a compact open subgroup. This implies that, when $A=\mathbb{Z}_d^n$ ($d$ odd), we recapture the result in \cite{gross2006hudson}. However, we notice that the above-mentioned definition of $S$-state is much more explicit than that of stabilizer state. We also observe that an explicit description (construction) of all $S$-states in any finite Abelian group is provided in \cite[Section 2.4]{nicJMPA}; see also \cite{hostens2005stabilizer,kaiblinger,kaiblinger09,nicola23} and Remark \ref{rem rem2} below. 

Although $Wf$ might take negative values, we will show that, under the assumption that the map $x\mapsto 2x$ is measure-preserving, a properly smoothed version of $Wf$ is always nonnegative (see \cite[Theorem 4.4.4]{grochenig_book} for the analogous result in $\R^d$). We refer to Section \ref{sec caratteri} below for the relevant terminology.
\begin{theorem}\label{th:Wigner smoothed is nonnegative}
    Let $A$ be a 2-regular LCA group that contains a compact open subgroup. Suppose that the map $x\mapsto 2x$ is measure-preserving.
    
    Let $G \subset A \times \widehat{A}$ be a compact subgroup containing a maximal compact open isotropic subgroup. Then, for every $f \in L^2(A)$, we have
    \begin{equation*}
         Wf \ast \chi_G (x,\xi) \geq 0 \quad \forall \, (x,\xi) \in A \times \widehat{A}.
    \end{equation*}
\end{theorem}
Hence, the maximal compact open isotropic subgroups of $A\times \widehat{A}$, together with their translates, provide a coarse-graining of the phase space compatible with the uncertainty principle and can be regarded as the analogue of \textit{quantum blobs}, as defined by M. de Gosson in $\R^d\times\R^d$ (see \cite{degosson_wigner}). 
    
We now consider the case where the map $x\mapsto 2x$ is \textit{not} measure-preserving. 
\begin{theorem}\label{thm mainteo2}
    Let $A$ be a $2$-regular LCA group that contains a compact open subgroup. Suppose that the map $x\mapsto 2x$ is not measure-preserving. Then, for every $f\in L^2(A)\setminus\{0\}$, $Wf(x,\xi)<0$ for some $(x,\xi)\in A\times \widehat{A}$. 
\end{theorem}
Section \ref{sec proofsmainresults} is devoted to the proof of these and some related results.  The proofs are based on a recent characterization \cite{nicola2023maximally} of the extremal functions for the generalized Wehrl entropy bound (for the concave function $\Phi(\tau)=\sqrt{\tau}$, $\tau\in [0,1]$; see Section \ref{sec caratteri} for terminology). Indeed,
in \cite{lieb} H. Lieb sketched an argument that shows that the problems of positivity of the Wigner distribution in $\R^d$ and the Wehrl entropy bounds are related \textit{at least at the formal level} (however, formula \cite[(1,26)]{lieb} does not hold in the whole space $L^2$, but only for functions in the so-called Feichtinger algebra --- see \cite{degosson_wigner} --- which prevents one from concluding in the desired generality).  We will see how to rigorously implement Lieb's idea following a technically different pattern.

In Section \ref{sec products} we specialize Theorem \ref{thm mainteo 1} to the case of groups that are (possibly infinite) sums of discrete Abelian groups or (possibly infinite) products of compact Abelian groups. Precisely, let $\{A_j:\ j\in I\}$ be a family of $2$-regular discrete groups, where $I$ is a nonempty set. Let $A=\bigoplus_{j\in I} A_j$. We will see that the functions that have the form in Theorem \ref{thm mainteo 1} (b), that is, the $S$-states of $A$, are precisely the functions of the type $\tilde{f}\otimes \delta$, where $\tilde{f}$ is an $S$-state of $\bigoplus_{j\in F} A_j$ for some finite subset $F\subset I$ and $\delta$ is the indicator function of $\{0\}$ in $ \bigoplus_{j\in I\setminus F} A_j$. 

This result has a natural interpretation in terms of infinite quantum spin systems, which represents one of the motivations for this note. Consider a composite system, whose subsystems are described by the Hilbert spaces $\ell^2(A_j)$, $j\in\mathbb{N}$, where $A_j$ is a finite Abelian group (that is, the classical configuration space of a qudit, or spin system). Then (for a suitable choice of the local ground states, see \cite{naaijkens} and Section \ref{sec quantumspin}), $\ell^2(A)$ can be interpreted as the Hilbert space of the composite system, and we conclude that $Wf\geq0$ everywhere if and only if $f$ is an $S$-state of a finite subsystem --- in particular, only local excitations occur. A similar result holds for the (possibly infinite) product of compact Abelian groups, with a similar interpretation for a different choice of local ground states; see Section \ref{sec quantumspin}. 

In view of the above-mentioned connections between the issue of positivity of the Wigner distribution and several deep problems in quantum mechanics, it is expected that the machinery developed in this note could similarly be applied to extend some of those results to infinite quantum spin systems --- we plan to address this question in a subsequent work.   

In Section \ref{sec examples} we will discuss some further examples, including $n$-adic groups (for an integer $n$ not necessarily prime) and solenoid groups.

\section{Notation and preliminaries}
\subsection{Notation}\label{sec notation}In the following, we will denote a locally compact Hausdorff Abelian (shortly, LCA) group with $A$. The dual group of $A$, that is, the LCA group of continuous homomorphisms $\xi \colon A \to U(1)$ (where $U(1) = \{z \in \C  \colon |z|=1 \}$ with the usual product), will be denoted as $\widehat{A}$ and the action of $\xi \in \widehat{A}$ on $x \in A$ as $\xi(x) \in \C$, $|\xi(x)|=1$. Both the group laws on $A$ and $\widehat{A}$ will be denoted additively.

The Haar measure on $A$ will be denoted as $dx$ and $|\Omega|$ denotes the measure of a (Borel) subset $\Omega\subset A$. 

Given a function $f \in L^1(A)$, its Fourier transform $\widehat{f} \colon A \to \C$ is defined as
\begin{equation*}
    \widehat{f}(\xi) = \int_{A} f(x) \overline{\xi(x)} \, dx, \quad \xi \in \widehat{A}.
\end{equation*}
The Fourier transform is then extended to $L^2(A)$, and on $\widehat{A}$ we choose the Haar measure $d\xi$ normalized in such a way that Plancherel's formula holds with constant 1, that is,
\begin{equation*}
    \int_A |f(x)|^2 \, dx = \int_{\widehat{A}} |\widehat{f}(\xi)|^2 \, d\xi.
\end{equation*}
We also make use of the pointwise Fourier inversion formula in the following form: if $f\in L^2(A)$, $\widehat{f}\in L^1(A)$ and $f$ is continuous, then, 
\[
f(x)=\int_{\widehat{A}} \widehat{f}(\xi)\xi(x)\, d\xi\qquad x\in A.
\]
This is a consequence of the Plancherel formula in the form of \cite[(31.18)]{hewitt70}. 

Given a subgroup $H\subset A$, we set 
\[
H^\bot=\{\xi\in\widehat{A}:\ \xi(x)=1\quad\forall x\in H\} \subset \widehat{A}
\]
for its orthogonal subgroup. If $H$ is compact and open in $A$ then $H^\bot$ is compact and open in $\widehat{A}$, $\widehat{\chi_H}=|H|\chi_{H^\bot}$ and $|H| |H^\bot|=1$; see, e.g., \cite[Example 4.4.9]{reiter00}. 

\subsection{On 2-regular LCA groups}
On an LCA group $A$, we consider the map $x \mapsto 2x \coloneqq x+x$. It is clear that it is a continuous group homomorphism $A \to A$. 
\begin{definition}\label{def 2regular} We say that an LCA group $A$ is \emph{2-regular} if the map $x\mapsto 2x$ is bijective and open, hence a topological isomorphism. In that case, we denote by $2^{-1}:A\to A$ its inverse.
\end{definition}
   We point out that if $A$ is 2-regular, then also $\widehat{A}$ is 2-regular, since $(2\xi)(x) = \xi(2x)$ and $(2^{-1}\xi)(x) = \xi(2^{-1}x).$ 

If $A$ is 2-regular, the map $\Omega\mapsto |2\Omega|$ ($\Omega$ Borel) is a Radon measure. Indeed, this is again a Haar measure, because $2(x+\Omega) = 2x+2\Omega$ and, therefore, there exists $\lambda_2 >0$ such that
\begin{equation}\label{eq:lambda_2}
    |2\Omega| = \lambda_2|\Omega| \quad \forall \, \Omega \subset A \quad \text{Borel}.
\end{equation}
We point out that the measure $\Omega\mapsto |2\Omega|$ is nothing more than the push-forward of the Haar measure $dx$ by $2^{-1}$. Therefore, given a measurable function $f \colon A \to \C$, summable or nonnegative, from the change-of-variable formula we have
\begin{equation}\label{eq:change of variable}
    \int_A f(2^{-1}x) \, dx = \lambda_2 \int_A f(y)  \, dy,
\end{equation}
and the same holds when integrating over $\widehat{A}$.
Indeed, the same argument applies to $\widehat{A}$ and, moreover, the constant $\lambda_2$ is the same. This can be seen by the Plancherel formula and the fact that 
\[
\widehat{f}(2^{-1}\cdot)=\lambda_2\widehat{f(2^{-1}\cdot)}.
\]
Suppose now that the 2-regular group $A$ contains a compact open subgroup $H$. Then $2 H \subset H$ and, moreover, $|H|<\infty$ because $H$ is compact and $|H|>0$ because $H$ is open. Therefore, it must be 
\[
0<\lambda_2 \leq 1.
\]
In the case $A$ contains a compact open subgroup, we can see that the doubling constant $\lambda_2$ is the same for $A$ and $\widehat{A}$ also by noticing that if $H$ is an open compact subgroup of $A$ then $H^{\perp}$ is also an open compact subgroup of $\widehat{A}$ and $(2H)^{\perp}=2^{-1}(H^{\perp})$.

\subsection{Phase-space analysis on LCA groups}\label{sec phasespace}

We recall some elementary facts about phase-space analysis in LCA groups. Let $A$ be an LCA group. Given $x \in A$, $\xi \in \widehat{A}$, we define the translation by $x$,
\begin{equation*}
    T_xf(y) = f(y-x), \quad y \in A,
\end{equation*}
and the modulation by $\xi$,
\begin{equation*}
    M_{\xi}f(y) = \xi(y)f(y), \quad y \in A,
\end{equation*}
where $f\in L^2(A)$. 

If $f,g \in L^2(A)$ we define the short-time Fourier transform \cite{grochenig1998aspects}, alias coherent state transform \cite{lieb_book}, of $f$ with window $g$, as
\begin{equation*}
    V_gf(x,\xi) = \langle M_{\xi}T_x g, f \rangle_{L^2} \quad (x,\xi) \in A \times \widehat{A},
\end{equation*}
where the scalar product $\langle \cdot,\cdot\rangle_{L^2}$ is considered linear in the second argument.
The function $V_gf$ is continuous and vanishes at infinity (see \cite[Proposition 2.1]{nicola2023maximally}). 

We have the Plancherel-type formula for $f_1,f_2,g_1,g_2\in L^2(A)$:
\[
\langle V_{g_1}f_1, V_{g_2}f_2\rangle_{L^2(A\times\widehat{A})}=\langle f_1,f_2 \rangle_{L^2(A)}\langle g_2,g_1 \rangle_{L^2(A)}.
\]
In particular, if $f,g\in L^2(A)$ and $\|g\|_{L^2(A)}=1$ we have 
\[
\|V_g f\|_{L^2(A\times\widehat{A})} =\|f\|_{L^2(A)}\qquad (\|g\|_{L^2(A)}=1),
\]
that is, $V_g:L^2(A)\to L^2(A)$ is a (not onto) isometry.  

We refer to \cite{grochenig1998aspects} for the proof of these results (see also \cite{grochenig_book} for the analogous results in $\R^d$ and \cite{feichtinger} for the case where $A$ is finite). 

Suppose now that $A$ is 2-regular. Then we can also define the \emph{Weyl--Heisenberg operators} $w(x,\xi):L^2(A)\to L^2(A)$:
\begin{equation*}
    w(x,\xi) = \overline{\xi(2^{-1}x)} M_{\xi}T_x.
\end{equation*}
These unitary operators are a symmetrized version of the phase-space shifts $M_\xi T_x$, and allow us to \textit{naturally} extend the definition of the Wigner distribution, given in \eqref{def wigner}, to any $2$-regular LCA group (we notice that a number of different proposals were also made for Wigner-type functions in LCA groups that are not 2-regular; see, e.g., \cite{kutyniok2003,mantoiu}). 
\begin{definition}\label{def wigner} Let $A$ be a $2$-regular LCA group. 
Given $f \in L^2(A)$, we define its \textrm{ambiguity function} as
\begin{equation}\label{eq:def_ambiguity}
    \A f(x,\xi) = \langle w(x,\xi)f,f\rangle_{L^2(A)} = \xi(2^{-1}x) V_ff(x,\xi) \quad (x,\xi) \in A \times \widehat{A}
\end{equation}
and its \textrm{Wigner distribution} as
\begin{equation}\label{eq:def_Wigner}
    Wf(x,\xi) =  \int_A f(x + 2^{-1}y) \overline{f(x-2^{-1}y)}\, \overline{\xi(y)} \, dy \quad (x,\xi) \in A \times \widehat{A}.
\end{equation}
\end{definition}
We observe that both the functions $\mathcal{A}f$ and $Wf$ are continuous.

From a direct computation, using the change-of-variable formula \eqref{eq:change of variable}, it is easy to see that
\begin{equation}\label{eq:relation Wigner STFT}
    Wf(x,\xi) = \lambda_2\, \xi(2x) V_{f^{\vee}}f(2x,2\xi) \quad (x,\xi) \in A \times \widehat{A},
\end{equation}
where $f^{\vee}$ is the reflection of $f$, that is, $f^{\vee}(x): = f(-x)$.

Also, denote by $\mathscr{F}$ the Fourier transform in $\widehat{A} \times A$, that is, 
\[
\mathscr{F}F(x,\xi)=\iint_{\widehat{A}\times A}  F(\eta,y)\overline{\xi(y)}\overline{\eta(x)}\, d\eta\,dy.
\]
Then the ambiguity and the Wigner transforms of $f$ are related by 
\begin{equation}\label{eq:relation Wigner and Ambiguity}
    Wf = \mathscr{F}\mathcal{U}\mathcal{A}f,
\end{equation}
where $\mathcal{U}$ is the rotation of a function $F$ on $A \times \widehat{A}$, that is, 
\begin{equation}\label{def U}
\mathcal{U}F(\xi,x) = F(x,-\xi);
\end{equation}
see \cite[Lemma 4.3.4]{grochenig_book} for the proof in $\R^d$, which extends easily to the LCA setting.

Following the same argument as in $\R^d$ (see \cite[Proposition 4.3.2]{grochenig_book}) one can prove the covariance property
\begin{equation}\label{eq:covariant property Wigner}
    W(T_yM_{\eta}f)(x,\xi) = W(x-y,\xi-\eta),\qquad x,y\in A,\ \xi,\eta\in\widehat{A},
\end{equation}
and Moyal's formula
\begin{equation}\label{eq:Moyal formula}
    \langle Wf,Wg\rangle_{L^2(A \times \widehat{A})}=|\langle f,g\rangle_{L^2(A)}|^2\qquad f,g\in L^2(A).
\end{equation}
Also, for every $f \in L^2(A)$ it holds
    \begin{equation}\label{eq wigner fourier}
        Wf(x,\xi) = W\widehat{f}(\xi,-x) \quad (x,\xi) \in A \times \widehat{A}.
    \end{equation}
    Again, this can be proved following the same pattern as in $\R^d$ (see \cite[Proposition 4.3.2]{grochenig_book}).

\subsection{Subcharacters of second degree, \texorpdfstring{$S$}{S}-states and Wehrl entropy}\label{sec caratteri}
In this section we recall some results from \cite{igusa68,nicola2023maximally,weil64} which will be useful in the following. 

Let $A$ be an LCA group. The standard symplectic structure of $A \times \widehat{A}$ is defined by the bicharacter $\sigma \colon (A \times \widehat{A}) \times (A \times \widehat{A}) \to U(1)$ given by
\begin{equation*}
    \sigma((x,\xi),(y,\eta)) = \xi(y) \overline{\eta(x)}, \quad (x,\xi),\,(y,\eta) \in A\times\widehat{A}.
\end{equation*}
A subgroup $G \subset A \times \widehat{A}$ is called \emph{isotropic} if $\sigma(z,w)=1$ for every $z,w \in G$. If $G$ is a compact open isotropic subgroup, then $|G| \leq 1$ (\cite[Proposition 3.5]{nicola2023maximally}). If $|G|=1$, then $G$ will be called \emph{maximal compact open isotropic subgroup}. It follows from \cite[Proposition 3.1, 3.5 and 3.6]{nicola2023maximally} that every such subgroup $G$ has the form 
\[
G=\{(x,\xi)\in A\times\widehat{A}: x\in H,\  \xi|_H= \beta(x)\} 
\]
for some (uniquely determined) compact open subgroup $H\subset A$ and continuous symmetric homomorphism $\beta:H\to \widehat{H}$.

We recall that, a continuous homomorphism $\beta \colon A \to \widehat{A}$ is called \emph{symmetric} if $\beta(x)(y) = \beta(y)(x)$ for every $x,y \in A$. We denote by $\mathrm{Sym}(A)$ the set of such homomorphisms, which is an Abelian group with the operation 
\[
(\beta + \beta')(x)(y) \coloneqq \beta(x)(y)\beta'(x)(y),
\qquad \beta,\beta'\in \mathrm{Sym}(A).
\]
We now recall from \cite{weil64} (see also \cite{reiter89}) the definition of character of second degree. 

\begin{definition}\label{def caratsecond}
     A \emph{character of second degree} associated with $\beta \in \mathrm{Sym}(A)$ is a continuous function $h \colon A \to U(1)$ such that
\begin{equation*}
    h(x+y) = h(x)h(y)\beta(x)(y), \quad \forall \, x,y \in A.
\end{equation*}
\end{definition}
We denote by $\mathrm{Ch}_2(A)$ the set of characters of the second degree, which is an Abelian group with the operation 
\[
(h+h')(x) \coloneqq h(x)h'(y),\qquad h,h'\in \mathrm{Ch}_2(A).
\]
Moreover, from \cite[Lemma 6]{igusa68} and \cite[page 146]{weil64} we have the following short exact sequence of Abelian groups:
\begin{equation}\label{eq:exact sequence}
    0 \to \widehat{A} \to \mathrm{Ch}_2(A) \to \mathrm{Sym}(A) \to 0.
\end{equation}
In particular, this implies that for every $\beta \in \mathrm{Sym}(A)$ there exists $h \in \mathrm{Ch}_2(A)$ associated with $\beta$ and that two $h,h' \in \mathrm{Sym}(A)$ associated with the same $\beta \in \mathrm{Sym}(A)$ differ for a character, i.e., $h=M_{\xi}h'$ for some $\xi \in \widehat{A}$. 

\begin{remark}\label{rem rem1}
    If $A$ is 2-regular, the function $A\ni x\mapsto \beta(2^{-1}x)(x)$ is, in fact, a character of second degree associated with $\beta\in{\rm Sym}(A)$, and therefore in that case all the characters of second degree associated with $\beta$ have the form $x\mapsto \beta(2^{-1}x)(x)\xi(x)$ for some $\xi\in\widehat{A}$. 
\end{remark} 
We also need the following definition from \cite[Definition 4.3]{nicola2023maximally}.
\begin{definition}\label{def subcharacter}
    Given a compact open subgroup $H \subset A$ and a continuous symmetric homomorphism $\beta \in \mathrm{Sym}(H)$, we say that a continuous function $\phi \colon A \to \C$ is a \emph{subcharacter of second degree} associated with $(H,\beta)$ if $\phi(x)=0$ for every $x \notin H$ and $\phi\vert_H$ is a character of second degree associated with $\beta$ (Definition \ref{def caratsecond}).
    \end{definition}
     This definition was inspired by that of \emph{subcharacter}, that is, a continuous function $\phi \colon A \to \C$ such that for some compact open subgroup $H \subset G$ we have $\phi(x) = 0$ for $x \notin H$ and $\phi\vert_H$ is a character of $H$ (see \cite[Definition 43.3]{hewitt70}). 

     We now introduce the notion of the $S$-state in $A$, that is, a function that has the form in Theorem \ref{thm mainteo 1} (b).

    \begin{definition}\label{def Sstate}
An \emph{$S$-state} in $A$ is a continuous function $f:A\to\C$ of the form 
\[
f(x)=c\,\phi(x-y)\qquad x\in A
\]
for some subcharacter of second degree $\phi$ of $A$, some $y\in A$ and $c\in \C\setminus\{0\}$. 
    \end{definition}
    \begin{remark}\label{rem rem2}
By definition, the $S$ states are precisely the functions that appear in Theorem \ref{thm mainteo 1} (b). However, under the assumption of this theorem, namely that $A$ is 2-regular and the doubling map $x\mapsto 2x$ is measure-preserving, the $S$-states can be described more explicitly. Indeed, if $H\subset A$ is a compact open subgroup, we have $2H\subset H$, and $|2H|=|H|$, hence $2H=H$, that is, $H$ is 2-regular. Hence, as a consequence of Remark \ref{rem rem1}, a subcharacter of second degree $\phi$ of $A$, with support $H$, necessarily has the form
\[
\phi|_H(x)=\beta(2^{-1}x)(x)\xi(x)\qquad x\in H,
\]
for some $\beta\in{\rm Sym}(H)$ and $\xi\in \widehat{H}$.

We also observe that regardless of whether $A$ is 2-regular or not, by \cite[Theorem 1.4]{nicJMPA} the $S$-states coincide with the so-called stabilizer states, as defined in \cite[Definition 1.4]{nicJMPA}.
    \end{remark}
We can finally recall some of the results from \cite{nicola2023maximally} that will be used later. The first result is \cite[Proposition 4.4]{nicola2023maximally}, which gives information on the transform $V_\phi \phi$ of a subcharacter of second degree $\phi$.

\begin{theorem}\label{thm thm26}
Let $\phi$ be a subcharacter of second degree of $A$ associated with some pair $(H,\beta)$, where $H\subset A$ is a compact open subgroup and $\beta\in{\rm Sym}(H)$ (Definition \ref{def subcharacter}).  Then $|H|^{-1}V_\phi \phi(x,\xi)$ is a subcharacter of second degree of $A\times \widehat{A}$ associated with the pair $(G,\beta')$, where $G$ is the maximal compact open isotropic subgroup
\[
G=\{(x,\xi)\in A\times\widehat{A}:\ x\in H,\  \xi|_H=\beta(x)\}
\]
 and the symmetric homomorphism $\beta'\in{\rm Sym}(G)$ is given by  
    \[
    \beta'(y,\eta)(x,\xi)=\overline{\eta(x)}\qquad (x,\xi),(y,\eta)\in G.
    \]
\end{theorem}

For the benefit of the reader, we also recall the following result from \cite[Theorem 5.2]{nicola2023maximally}.
\begin{theorem}\label{th:equivalence for STFT}
    Let $f,g \in L^2(A)$ be normalized. The following facts are equivalent:
    \begin{itemize}
        \item $|\{z \in A \times \widehat{A} \, \colon \, V_gf(z) \neq 0\}|=1$;
        \item $|V_gf|=\chi_S$, where $S \subset A \times \widehat{A}$ is a coset of a maximal compact open isotropic subgroup of $A \times \widehat{A}$;
        \item $f$ and $g$ are $S$-states and $f = cM_{\xi}T_xg$ for some $(x,\xi) \in A \times \widehat{A}$ and $c \in U(1)$.
    \end{itemize}
\end{theorem}
Finally, the following estimate was proved in \cite[Theorem 6.3.1]{grochenig1998aspects}, while the characterization of the extremizers was obtained in \cite[Theorem 7.2]{nicola2023maximally} (see also \cite{lieb} for the case $A=\R$).
\begin{theorem}\label{th:Wehrl entropy}
    For every $f,g \in L^2(A)$ we have
    \begin{equation*}
        \|V_gf\|_{L^1(A\times\widehat{A})} \geq \|f\|_{L^2(A)}\|g\|_{L^2(A)}.
    \end{equation*}
    If $f,g \neq 0$, equality is achieved if and only if both are $S$-states and $f=cM_{\xi}T_xg$ for some $x \in A$, $\xi \in \widehat{A}$ and $c \in \C \setminus \{0\}$.
\end{theorem}
The above estimate can be regarded as an instance of a Wehrl entropy bound, that is, the problem of minimizing the generalized entropy 
\[
\iint_{A\times\widehat{A}}\Phi(|V_g f|^2)\, dx\, d\xi
\]
over normalized $f,g\in L^2(A)$, where $\Phi \colon[0,1]\to\R$ is a concave function. Indeed, $\Phi(\tau)=\sqrt{\tau}$ in Theorem \ref{th:Wehrl entropy}. The sharp lower bound (in LCA groups containing a compact open subgroup), for a general concave function $\Phi$, was recently obtained in \cite{nicJMPA}, but in the following we will only use the special case of  Theorem \ref{th:Wehrl entropy}; see \cite{lieb1978,lieb} for context when $A=\R^d$; see also \cite{NRT2024} and the references therein for recent developments.

\section{Main results and proof}\label{sec proofsmainresults}
Before diving into the proof of our main theorem, we need the following lemma, which can be seen as a refined version of the Fourier inversion formula.
\begin{lemma}\label{lem:Fourier inversion}
    Let $A$ be an LCA group, $f \in L^2(A)$ continuous such that $\widehat{f} \geq 0$. Then,
    \begin{equation*}
        \int_{\widehat{A}} \widehat{f}(\xi) \, d\xi = f(0).
    \end{equation*}
\end{lemma}
\begin{proof}
    For every $\varphi \in L^2(A)$, from the Plancherel formula we have
    \begin{equation*}
        \int_A f(x) \overline{\varphi(x)} \, dx = \int_{\widehat{A}} \widehat{f}(\xi) \overline{\widehat{\varphi}(\xi)} \, d\xi.
    \end{equation*}
    Consider a basis of compact neighborhoods $\{U\}$ of 0 and a corresponding approximate identity $\psi_U$ (see \cite[Proposition 2.42]{follandbook}). One can suppose that each $U$ is symmetric, that is, $x \in U \implies -x \in U$. Let $\varphi_U = \psi_U \ast \psi_U$. Then $\varphi_U$ is still an approximate identity associated with the basis of neighborhoods $\{U+U\}$, which means that $\varphi_U \geq 0$, $\int_A \varphi_U =1$, and $\varphi_U \in C(A)$ is supported in the compact subset $U+U$. Moreover, $\widehat{\varphi_U} = |\widehat{\psi_U}|^2 \geq 0$ and $\widehat{\varphi_U} \to 1$ uniformly in compact subsets of $\widehat{A}$ (see \cite[Lemma 4.46]{follandbook}). So, from the Plancherel formula with $\varphi=\varphi_U$ and recalling that both $\widehat{f}$ and $\widehat{\varphi_U}$ are nonnegative, for every compact $K \subset \widehat{A}$ we have
    \begin{equation*}
        \int_K \widehat{f}(\xi) \widehat{\varphi_U}(\xi) \, d\xi \leq \int_{\widehat{A}} \widehat{f}(\xi) \widehat{\varphi_U}(\xi) \, d\xi = \int_A f(x) \varphi_U(x) \, dx.
    \end{equation*}
    Letting $U \to \{0\}$ (in the sense of nets) we deduce that
    \begin{equation*}\label{eq:estimate L^1}
        \int_K \widehat{f}(\xi) \, d\xi \leq f(0).
    \end{equation*}
    Since $\widehat{f} \in L^2(\widehat{A})$, the set $\{\xi\in\widehat{A}:\ \widehat{f}(\xi) \neq 0\}$ is contained in a $\sigma$-compact subset of $\widehat{A}$ (see \cite[Page 44]{follandbook}) and so we have $\{\xi\in\widehat{A}:\ \widehat{f}(\xi) \neq 0\}\subset \cup_n K_n$ for some sequence $K_n \subset K_{n+1}$ of compact sets. Then, using monotone convergence in the previous estimate, with $K=K_n$, we conclude that $\widehat{f} \in L^1(A)$. We are now in a position to apply the usual Fourier inversion formula (since $f\in L^2(A)$ is continuous, and $\widehat{f}\in L^1(A)$; see Section \ref{sec notation}), which gives the desired result.
\end{proof}

We are now in a position to prove our main results. We start by proving Theorem \ref{thm mainteo 1}.
\begin{proof}[Proof of Theorem \ref{thm mainteo 1}] 
    We start with the implication $(b) \implies (a)$. First, we notice that if $f$ has the form in (b), then $f^{\vee} = e^{i\theta}M_{\xi_0} T_{x_0}f$ for some $(x_0,\xi_0) \in A \times \widehat{A}$ and some $\theta \in \R$ (recall $f^{\vee}(x):=f(-x))$. Indeed, let $f = cT_y\phi$ where $\phi$ is a subcharacter of second degree of $A$ whose support is some compact open subgroup $H\subset A$; hence $\phi|_H\in{\rm Ch}_2(H)$.   Then, $\phi^{\vee}|_H$ is still a character of second degree of $H$, and therefore $\phi^{\vee} = M_{\xi_0}\phi$ for some $\xi_0 \in \widehat{A}$ (because every continuous character of $H$ extends to a continuous character of $A$), which implies $f^{\vee} = \xi_0(y) M_{\xi_0}T_{-2y} f$.
    
    Hence, assuming (without loss of generality) $\|f\|_{L^2}=1$, from Theorem \ref{th:equivalence for STFT} we see that $|V_{f^{\vee}}f| = \chi_S$, where $S$ is the coset of a maximal compact open isotropic subgroup of $A \times \widehat{A}$, which implies, in particular, that $|S|=1$. Hence, using the relation \eqref{eq:relation Wigner STFT} between the coherent-state transform and the Wigner distribution, the change-of-variable formula \eqref{eq:change of variable}, and the fact that $\lambda_2=1$ in \eqref{eq:change of variable} by assumption,  we have
    \begin{align*}
        \int_{A \times \widehat{A}} |Wf(x,\xi)| \, dxd\xi&= \int_{A \times \widehat{A}}|V_{f^{\vee}}f(2x,2\xi)| \, dxd\xi\\
        &= \int_{A \times \widehat{A}} |V_{f^{\vee}}f(x,\xi)| \, dx d\xi = |S| = 1.
    \end{align*}
    In particular, $Wf \in L^1(A \times \widehat{A})$. Moreover, recalling the relation \eqref{eq:relation Wigner and Ambiguity} between the Wigner transform and the ambiguity function of $f$, since $\A f \in L^2(A \times \widehat{A})$ and $\A f$ is continuous, we can apply the Fourier inversion formula (see Section \ref{sec notation}) to obtain
    \begin{equation*}
        \int_{A \times \widehat{A}} Wf(x,\xi) \, dx d\xi = \mathcal{U}\A f(0,0) = 1.
    \end{equation*}
    This fact, combined with the previous result, implies $Wf \geq 0$.

    We now move to the proof of $(a) \implies (b)$. Since $\mathcal{U}\A f \in L^2(\widehat{A} \times A)$ is continuous and $Wf = \mathscr{F}\mathcal{U}\A f \geq 0$, from Lemma \ref{lem:Fourier inversion} we have
    \begin{equation*}
        \int_{A \times \widehat{A}} Wf(x,\xi) \, dx d\xi = \mathcal{U}\A f(0,0) = \|f\|_{L^2(A)}^2.
    \end{equation*}
    On the other hand, using again the change-of-variable formula \eqref{eq:change of variable} ($\lambda_2=1$), and the relation \eqref{eq:relation Wigner STFT}, we also have
    \begin{align*}
        \int_{A \times \widehat{A}} Wf(x,\xi) \, dx d\xi &= \int_{A \times \widehat{A}} |Wf(x,\xi)| \, dx d\xi \\
        &= \int_{A \times \widehat{A}} |V_{f^{\vee}}f(x,\xi)| \, dx d\xi \geq \|f\|_{L^2(A)}^2,
    \end{align*}
    where in the last step we used Theorem \ref{th:Wehrl entropy}. Hence, the last inequality is an equality, and again from Theorem \ref{th:Wehrl entropy} we then conclude that $f$ (and $f^{\vee})$ has the form in (b).
\end{proof}

This simple corollary is the result from Theorem \ref{thm mainteo 1} specialized for compact connected groups.
\begin{corollary}\label{cor:A connected compact}
   Let $A$ be a 2-regular LCA group. Assume that $A$ is compact and connected. Then, the following facts are equivalent for $f \in L^2(A)$.
    \begin{itemize}
        \item[(a)] $Wf(x,\xi)\geq 0$ for every $(x,\xi)\in A\times\widehat{A}$;
        \item [(b)] $f=c\, \xi$ for some $\xi \in \widehat{A}$ and some $c \in \C$.
    \end{itemize}
\end{corollary}
\begin{proof}
    Since $A$ is compact, $0<|A| < \infty$, and moreover $2A = A$ because $A$ is 2-regular. Therefore, $\lambda_2=1$ in \eqref{eq:lambda_2}, that is,  the doubling map $x\mapsto 2x$ is measure-preserving. By Theorem \ref{thm mainteo 1}, $Wf \geq 0$ everywhere if and only if $f=cT_{y}\phi$ for some subcharacter of second degree $\phi$ and some $c\in\C$. Hence we just need to prove that, when $A$ is compact and connected, subcharacters of second degree are actually characters.
    
    First of all we notice that if $A$ is connected then it has no proper open subgroups. In fact, if $H \subset A$ is an open subgroup, then it is also closed (\cite[Theorem 5.5]{hewitt63}) and therefore $H=A$. In particular, this implies that every subcharacter of second degree is actually a character of second degree.
    
    Consider now a symmetric homomorphism $\beta \colon A \to \widehat{A}$. Being $A$ compact, its dual $\widehat{A}$ is discrete and therefore $\{0\} \subset \widehat{A}$ is an open subgroup. Since $\beta$ is a continuous homomorphism, ${\rm Ker} \beta \subset A$ is an open subgroup of $A$, and therefore $\beta\equiv 1$, which implies $\mathrm{Ch}_2(A)=\widehat{A}$ or, in other words, that the only characters of second degree are the characters.
\end{proof}
The following result characterizes the $S$-states in terms of their Wigner distribution and will be used in the proof of Theorem \ref{th:Wigner smoothed is nonnegative}. 
\begin{proposition}\label{prop:Wigner of S-states}
Under the same assumptions of Theorem \ref{thm mainteo 1}, an $L^2$-normalized function $f:A \to \C$ is an $S$-state if and only if $Wf=\chi_S$, where $S \subset A \times \widehat{A}$ is a coset of a maximal compact open isotropic subgroup of $A \times \widehat{A}$. Moreover, for every subset $S$ of this kind, there exists an $S$-state $f$ such that $Wf = \chi_S$.
\end{proposition}
\begin{proof}
    We start supposing $f$ is an $L^2$-normalized $S$-state. Hence, $f = cT_y\phi$, where $c \in \C \setminus\{0\}$, $y \in A$ and $\phi$ is a subcharacter of second degree associated with some pair $(H,\beta)$ (Definition \ref{def subcharacter}). Since we are supposing that $f$ is normalized, we have $1 = \|f\|_2^2 = |c|^2 \|\phi\|_2^2 = |c|^2|H|$. We will prove the desired result by passing through the ambiguity function, using the formula \eqref{eq:relation Wigner and Ambiguity}. 
    
    From a direct computation, it is easy to see that
    \begin{equation*}
        V_ff(x,\xi) =\overline{\xi(y)}|H|^{-1}V_\phi \phi(x,\xi).
    \end{equation*}
    Hence, recalling \eqref{eq:def_ambiguity} we have
    \begin{equation*}
        \A f(x,\xi) = \overline{\xi(y)}\xi(2^{-1}x) |H|^{-1}V_\phi \phi(x,\xi).
    \end{equation*}
    From Theorem \ref{thm thm26} we know that $|H|^{-1}V_\phi \phi(x,\xi)$ is a subcharacter of second degree associated (in the sense of Definition \ref{def subcharacter}) with the maximal compact open isotropic subgroup 
\begin{equation}\label{eq G}
G=\{(x,\xi)\in A\times\widehat{A}: x\in H,\  \xi|_H=\beta(x)\} \subset A \times \widehat{A},
\end{equation}
hence $|G|=1$, and the symmetric homomorphism $\beta':G\to \widehat{G}$ given by  
    \[
    \beta'(y,\eta)(x,\xi)=\overline{\eta(x)}\qquad (x,\xi),(y,\eta)\in G.
    \]
Let us now prove that the function $(x,\xi) \mapsto\overline{\xi(2^{-1}x)}$, restricted to $G$, is also a subcharacter of the second degree associated with $(G,\beta')$.
    
Indeed, the doubling map $x\mapsto 2x$ is measure-preserving in $A$, hence also in $\widehat{A}$ and in $A\times\widehat{A}$. As a consequence,  we have $G = 2G$ because $G$ and $2G$ are open, $2G\subset G$ and $|2G|=|G|$.  Hence, for every  $(x,\xi),(y,\eta) \in G$ we have
    \begin{equation*}
        \overline{(\xi+\eta)(2^{-1}x+2^{-1}y)} = \overline{\xi(2^{-1}x)}\,\overline{\eta(2^{-1}y)}\,\overline{\xi(2^{-1}y)\eta(2^{-1}x)},
    \end{equation*}
    but since also $(2^{-1}x,2^{-1}\xi),(2^{-1}y,2^{-1}\eta) \in G$, from the isotropy of $G$ we obtain
    \begin{equation*}
        (2^{-1}\xi)(2^{-1}y) = (2^{-1}\eta)(2^{-1}x) \implies \xi(2^{-1}y) = \eta(2^{-1}x),
    \end{equation*}
    which implies
    \begin{equation*}
        \overline{(\xi+\eta)(2^{-1}x+2^{-1}y)} = \overline{\xi(2^{-1}x)}\,\overline{\eta(2^{-1}y)}\,\beta'(y,\eta)(x,\xi)
    \end{equation*}
    as claimed. 
    
    From the exact sequence \eqref{eq:exact sequence} we therefore deduce that the functions $\xi(2^{-1}x)$ and $|H|^{-1}V_\phi \phi (x,\xi)$, restricted to  $G$, differ by the multiplication by a character of $G$, which extends to a character of $A\times \widehat{A}$. As a consequence, for some $(y_0,\eta_0)\in A\times \widehat{A}$, 
    \[
\mathcal{A}f(x,\xi)=\overline{\xi(y_0)}\eta_0(x) \chi_G(x,\xi)\qquad (x,\xi)\in A\times \widehat{A}.
    \]
Then (cf. \eqref{def U})
\[
\mathcal{U}\mathcal{A}f(\xi,x)=\xi(y_0)\eta_0(x) \chi_G(x,-\xi).
\]
By \eqref{eq:relation Wigner and Ambiguity} we obtain that 
\[
Wf(x,\xi)=\chi_G(x-y_0,\xi-\eta_0),
\]
where we used the fact that the Fourier transform of the function $(\xi,x)\mapsto \chi_G(x,-\xi)$ is the indicator function of the set 
\[
\{(\xi,x)\in \widehat{A}\times A:\ (x,-\xi)\in G\}^\bot= G.
\]
Let us prove this last equality. The inclusion $\supset$ is clear because $G$ is isotropic. The equality then follows because both sides are open subgroups of $A\times\widehat{A}$ of measure $1$ (see Section \ref{sec notation}).  

To prove the converse, suppose we have $f \in L^2(A)$ with $\|f\|_{L^2}=1$, such that $Wf=\chi_S$, where $S = z + G\subset A \times \widehat{A}$ is a coset of a maximal compact open isotropic subgroup $G \subset A \times \widehat{A}$. Then $G$ has the form in \eqref{eq G} (see Section \ref{sec caratteri}) 
for some compact open subgroup $H\subset A$ and some $\beta\in{\rm Sym}(H)$ (Definition \ref{def caratsecond}). It follows from the exact sequence \eqref{eq:exact sequence} that there exists some $h\in {\rm Ch}_2(H)$ associated with $\beta$. Consider the subcharacter of second degree $\phi$ defined by $\phi(x)=0$ for $x\in A\setminus H$ and $\phi(x)=h(x)$ for $x\in H$. Then $\phi$ is associated with the pair $(H,\beta)$ in the sense of Definition \ref{def subcharacter}, and it
follows from the first part of the proof that $|H|^{-1}W\phi=\chi_{S'}$, where $S'$ is some coset of $G$. Using the covariance property \eqref{eq:covariant property Wigner} one can find $(x_0,\xi_0) \in A \times \widehat{A}$ such that $|H|^{-1}W(T_{x_0}M_{\xi_0}\phi)=\chi_S$. Then, by \eqref{eq:relation Wigner and Ambiguity} and since two $L^2$-normalized functions have the same ambiguity function if and only if they differ by a phase factor (see \cite[Proposition 2.3]{nicola2023maximally}), we conclude that $f = cT_{x_0}M_{\xi_0}\phi$ for some $c \in \C \setminus \{0\}$. Since $M_{\xi_0}\phi$ is still a subcharacter of second degree, $f$ is an $S$-state. 

Finally, since the coset $S$ from which we started was arbitrary, the last part of the statement is proved.
\end{proof}

\begin{proof}[Proof of Theorem \ref{th:Wigner smoothed is nonnegative}]
    We start by supposing that $G$ is a maximal compact open isotropic subgroup of $A \times \widehat{A}$. We have
    \begin{equation*}
        Wf \ast \chi_G(x,\xi) = \int_{A \times \widehat{A}} Wf(y,\eta) \chi_G(x-y,\xi-\eta) \, dy d\eta = \langle Wf, \chi_S \rangle_{L^2(A \times \widehat{A})},
    \end{equation*}
    where $S = (x,\xi)+G$. From Proposition \ref{prop:Wigner of S-states} we know that there exists an $S$-state $g$ (which depends on $(x,\xi)$) such that $Wg=\chi_S$. Therefore, from Moyal's formula \eqref{eq:Moyal formula} we deduce
    \begin{equation*}
        Wf \ast \chi_G(x,\xi) = \langle Wf, Wg \rangle_{L^2(A \times \widehat{A})} = |\langle f,g \rangle_{L^2(A)}|^2 \geq 0.
    \end{equation*}
    Clearly, the same conclusion holds if $G$ is replaced by one of its cosets in $A\times\widehat{A}$.

    Now, assume $G$ as in the statement. Then, $G = \bigcup_{\lambda \in \Lambda} S_{\lambda}$, for some set of indices $\Lambda$, where the $S_{\lambda}$'s are the pairwise disjoint cosets, in $G$, of a maximal compact open isotropic subgroup of $A\times \widehat{A}$ (contained in $G$). Since the $S_\lambda$'s are open and $G$ is compact, $\Lambda$ is finite, and  $\chi_G = \sum_{\lambda \in \Lambda} \chi_{S_{\lambda}}$. The result then follows from the first part of the proof and the linearity of the convolution.
\end{proof}
We now switch to the case where the map $x \mapsto 2x$ is not measure-preserving.
\begin{proof}[Proof of Theorem \ref{thm mainteo2}] Let $f\in L^2(A)$, and suppose that $Wf\geq0$ everywhere. Using Lemma \ref{lem:Fourier inversion}, we have
\begin{equation*}
    \int_{A \times \widehat{A}} Wf(x,\xi) \, dx d\xi = \|f\|_{L^2(A)}^2.
\end{equation*}
On the other hand, arguing as in the  proof of Theorem \ref{thm mainteo 1}, hence using \eqref{eq:relation Wigner STFT}, where now $\lambda_2<1$ by assumption, 
\begin{align*}
     \int_{A \times \widehat{A}} Wf(x,\xi) \, dx d\xi &= \int_{A \times \widehat{A}} |Wf(x,\xi)| \, dx d\xi \\
    &= \lambda_2 \int_{A \times \widehat{A}} |V_{f^{\vee}}f(2x,2\xi)| \, dx d\xi \\
    &= \lambda_2^{-1} \int_{A \times \widehat{A}} |V_{f^{\vee}}f(x,\xi)| \, dx d\xi \\
    &\geq \lambda_2^{-1} \|f\|_{L^2(A)}^2,
\end{align*}
where in the last step we used Theorem \ref{th:Wehrl entropy}. Since $\lambda_2 < 1$, this implies $f=0$.
\end{proof}

\section{Product groups}\label{sec products}
As an interesting consequence of Theorem \ref{thm mainteo 1}, in this section we are going to consider groups of the kind $A = \bigoplus_{j \in I}A_j$ or $A = \prod_{j \in I} A_j$, where $\{A_j\}_{j \in I}$ is a (possibly infinite) family of, respectively, discrete or compact 2-regular groups. We start by considering the discrete case.
\subsection{Direct sum of discrete groups}
Let $I$ be a nonempty set of indices and $\{A_j\}_{j \in I}$ a family of discrete 2-regular groups. We consider the direct sum of these groups, that is, the group
\begin{equation*}
    A=\bigoplus_{j \in I} A_j
\end{equation*}
consisting of all elements $x = (x_j)_{j \in I}$ in the Cartesian product of all $A_j$'s such that $x_j=0$ for all but finitely many $j \in I$, with the sum defined componentwise. We equip it with the discrete topology, so that it becomes an LCA group. We take the counting measure as Haar measure on $A$. We recall that the dual group of $A$ is the direct product of the $\widehat{A_j}$'s (equipped with the product topology), that is, \[
\widehat{\bigoplus_{j \in I}A_j} = \prod_{j \in I} \widehat{A_j}
\]
(see \cite[Theorem 23.22]{hewitt63}), which is a compact group of measure 1, as a consequence of the Plancherel formula.

In the following, the index set $I$ will be divided as $I=I'\cup I''$, for disjoint and nonempty subsets $I'$ and $I''$. This induces corresponding splittings (topological isomorphisms) $A\simeq A'\times A''$, and $\widehat{A}\simeq \widehat{A'}\times \widehat{A''}$, $A\times\widehat{A}\simeq A'\times A''\times \widehat{A'}\times \widehat{A''}$. We will consider tensor products $f\otimes g$, for $f\in \ell^2(A')$, $g\in \ell^2(A'')$ and similarly for functions in $\widehat{A}$, or in $A\times\widehat{A}$. 
\begin{theorem}\label{th:direct sum}
    Let $\{A_j\}_{j \in I}$ be a nonempty family of discrete 2-regular groups, and let $A = \bigoplus_{j \in I}A_j$ be its direct sum with the discrete topology and the counting measure. Then, the following are equivalent for $f \in \ell^2(A) \setminus \{0\}$.
    \begin{itemize}
        \item[(a)] $Wf(x,\xi) \geq 0$ for all $(x,\xi) \in A \times \widehat{A}$.
        \item[(b)] There exist $F \subset I$ finite and an $S$-state $\widetilde{f}$ for $\bigoplus_{j \in F}A_j$ such that 
        \begin{equation}\label{eq:S-state in direct sum group}
            f = \widetilde{f} \otimes \delta,
        \end{equation}
        where $\delta \colon \bigoplus_{j \in I \setminus F} A_j \to \C$ is such that $\delta(y)=1$ if $y=0$ and $\delta(y)=0$ otherwise.
    \end{itemize}
\end{theorem}
\begin{proof}
    We observe that $A$ is 2-regular since the map $x \mapsto 2x$ can be inverted componentwise and  $A$ has the discrete topology. Moreover, the compact open subgroup $\{0\}$ is mapped on itself, therefore $\lambda_2=1$. Hence, from Theorem \ref{thm mainteo 1} we know that $Wf\geq0$ everywhere if and only if $f$ is an $S$-state. So, to prove our theorem, we just need to prove that all $S$-states of $A$ are of the kind \eqref{eq:S-state in direct sum group}.

    First of all, we notice that an $S$-state $f$ is necessarily supported in a compact, and hence finite set $K\subset A$. Therefore, there exists $F\subset I$ finite such that if $x=(x_j)_{j\in I}\in K$ then $x_j=0$ for $j\in I\setminus F$. Consider the splitting 
    \[
    A\simeq \bigoplus_{j\in F} A_j \times \bigoplus_{j\in I\setminus F} A_j,
    \]
    and correspondingly write $x=(x',x'')$ for $x\in A$. Then $f$
    has necessarily the form in \eqref{eq:S-state in direct sum group}, with $\tilde{f}(x'):=f(x',0)$; in particular $\widetilde{f} \in \ell^2(\bigoplus_{j \in F}A_j)$. So, we just need to prove that $f$ is an $S$-state of $A$ if and only if $\widetilde{f}$ is also an $S$-state of $\bigoplus_{j \in F}A_j$. According to the splitting induced in $A\times\widehat{A}$ 
     we have 
    $V_f f = V_{\widetilde{f}}\widetilde{f} \otimes V_{\delta}\delta$, and therefore 
    \[
    \{V_f f\not=0\}=\{V_{\widetilde{f}}\widetilde{f}\not=0\}\times \{ V_{\delta}\delta\not=0 \}.
    \]
     From a direct computation, one can see that $V_{\delta}\delta(y,\eta)=1$ if $y=0$ and is 0 otherwise, which implies 
    \[
    |\{ V_{\delta}\delta\not=0 \}|=|\{0\} \times \prod_{j \in I\setminus F}\widehat{A_j} |=|\prod_{j \in I\setminus F}\widehat{A_j}|=1.
    \]
    Therefore 
    \[
    |\{V_f f\not=0\}| = |\{V_{\widetilde{f}}\widetilde{f}\not=0\}|.
    \]
    From Theorem \ref{th:equivalence for STFT} we conclude that $f$ is an $S$-state if and only if also $\widetilde{f}$ is an $S$-state.
\end{proof}

\subsection{Direct product of compact groups}
In this section, we consider a family $\{A_j\}_{j \in I}$ of 2-regular compact Hausdorff Abelian groups, each with normalized Haar measure, and their direct product
\begin{equation*}
    A = \prod_{j \in I}A_j,
\end{equation*}
which is a compact Hausdorff Abelian group with Haar measure given by the Radon product of the Haar measures on the $A_j$'s (see \cite{folland_real_analysis}, in particular Section 7.4 for the definition of the Radon product of two Radon measures and Exercise 6 of Chapter 11 for the Haar measure on the direct product of compact groups). We observe that the dual group of $A$ is the direct sum of the $\widehat{A_j}$' s, that is, $\bigoplus_{j \in I}\widehat{A_j}$ equipped with with the discrete topology (see \cite[Theorem 23.21]{hewitt63}), where each $\widehat{A_j}$ is also discrete. 

Using the duality given by the Fourier transform, we can obtain the analogue of Theorem \ref{th:direct sum} but for the direct product of compact groups. Before doing so, we need the following lemma. 
\begin{lemma}\label{lem:f and f^ are S-states}
    Let $A$ be an LCA group. A function $f \in L^2(A)$ is an $S$-state for $A$ if and only if $\widehat{f} \in L^2(\widehat{A})$ is an $S$-state for $\widehat{A}$.
\end{lemma}
\begin{proof}
    From the easily verified relation
    \begin{equation*}
        V_{\widehat{f}}\widehat{f}(\xi,x) = \overline{\xi(x)} V_ff(-x,\xi) \quad (x,\xi) \in A \times \widehat{A}
    \end{equation*}
    we see that $|\{V_{\widehat{f}}\widehat{f} \neq 0\}| = |\{V_ff \neq 0\}|$, but from the characterization given by Theorem \ref{th:equivalence for STFT} this implies that $f$ is an $S$-state for $A$ if and only if $\widehat{f}$ is an $S$-state for $\widehat{A}$.
\end{proof}
\begin{theorem}\label{th:direct product}
    Let $\{A_j\}_{j \in I}$ be a family of 2-regular compact Hausdorff Abelian groups. Then, the following facts are equivalent for a function $f \in L^2(A) \setminus \{0\}$.
    \begin{itemize}
        \item[(a)] $Wf(x,\xi) \geq 0$ for all $(x,\xi) \in A \times \widehat{A}$.
        \item[(b)] There exist a finite subset $F \subset I$ and an $S$-state $\widetilde{f}$ for $\prod_{j \in F}A_j$ such that
        \begin{equation}\label{eq:S-states on direct product}
            f = \widetilde{f}\otimes 1, \quad x \in A,
        \end{equation}
        where $1$ denotes the function $\prod_{j \in I\setminus F}A_j\to \C$ identically equal to $1$.
    \end{itemize}
\end{theorem}
\begin{proof}
     By \eqref{eq wigner fourier}, $Wf \geq 0$ everywhere if and only if $W \widehat{f} \geq 0$ everywhere. Now, $\widehat{f} \in L^2(\widehat{A})$ and $\widehat{A}=\bigoplus_{j\in I}\widehat{A_j}$, so we can apply Theorem \ref{th:direct sum} to conclude that (a) holds true if and only if there exist $F \subset I$ finite and an $S$-state $g$ for  $\bigoplus_{j \in F} \widehat{A_j}$ such that $\widehat{f}= g \otimes \delta$. Taking the inverse Fourier transform we obtain that $f = \widetilde{f} \otimes 1$, with $g=\widehat{\widetilde{f}}$, and $\widetilde{f}$ is an $S$-state for $\prod_{j \in F}A_{j}$ if and only if $g$ is an $S$-state for $\bigoplus_{j \in F} \widehat{A_j}$, by Lemma \ref{lem:f and f^ are S-states}.
    
\end{proof}

\section{Infinite quantum spin systems and thermodynamic limit}\label{sec quantumspin}
In this section we show that Theorem \ref{th:direct sum} admits an interpretation in terms of the thermodynamic limit of a composite system where every subsystem is represented by the Hilbert space $\ell^2(A_j)$, where $A_j$ is a discrete 2-regular Abelian group. When each $A_j$ is finite, we end up with a (possibly infinite)  quantum spin system (see \cite{naaijkens}, in particular Sections 3.1 and 3.2). A similar interpretation of Theorem \ref{th:direct product} is discussed in the following.

According to von Neumann's construction (see, e.g., \cite[Section 4.5]{baez_segal_zhou}) of the (incomplete) infinite tensor product of Hilbert spaces --- in this case, $\ell^2(A_j)$ --- we need to select a \emph{ground state}, that is, a distinguished normalized vector in each space. For every $\ell^2(A_j)$, we choose as ground state the function $\delta_j \colon A_j \to \C$ given by $\delta_j(x)=1$ if $x=0$ and 0 otherwise. As ground state for $\ell^2(A)$, with \[
A=\bigoplus_{j\in I}A_j,
\]
we choose the function $\delta \colon A \to \C$ given by $\delta(x)=1$ if $x=0$ and 0 otherwise. Then, we have
\begin{equation*}
    (\ell^2(A),\delta) \simeq \bigotimes_{j \in I} (\ell^2(A_j),\delta_j),
\end{equation*}
where $\otimes_{j \in I}$ denotes the grounded tensor product. Exactly, with the symbol $\simeq$ we mean,  according to \cite[Definition on page 126]{baez_segal_zhou}, that for every $\Delta \subset I$ finite, there exists a linear isometric map \[
\tau_{\Delta} \colon \bigotimes_{j \in \Delta} \ell^2(A_j) \to \ell^2(A)
\]
such that
\begin{itemize}
    \item if $\Delta' \supset \Delta$ is another finite set of indices and if $f_j = \delta_j$ for $j \in \Delta'\setminus  \Delta$, then $\tau_{\Delta}\left( \otimes_{j \in \Delta} f_j\right) = \tau_{\Delta'}\left( \otimes_{j \in \Delta'} f_j\right)$;
    \item $\tau_{\Delta}\left( \otimes_{j \in \Delta} \delta_j\right) = \delta$;
    \item the union of the ranges of all the $\tau_{\Delta}$'s is dense in $\ell^2(A)$.
\end{itemize}
These maps $\tau_{\Delta}$ are given concretely by defining the function $\tau_{\Delta} \left( \otimes_{j \in \Delta} f_j\right)\in \ell^2(A)$ as
\begin{equation*}
    A \ni x \mapsto\prod_{j \in \Delta}f_j(x_j) \prod_{j \in I \setminus \Delta} \delta_j(x_j), \quad x = (x_j)_{j\in A} \in A,
\end{equation*}
which is well defined since in the second product only a finite number of terms is different from 1. The fact that the union of the ranges of all the $\tau_{\Delta}$'s is dense in $\ell^2(A)$ can be checked by observing that functions in $\ell^2(A)$ can be approximated by functions with finite support and arguing as in the proof of Theorem \ref{th:direct sum}. 

This realization of $(\ell^2(A),\delta)$ as a grounded tensor product is compatible with the action of the Weyl--Heisenberg operators in the subsystems and in the entire system. Indeed, consider $x = (x_j)_{j\in A} \in A$ and $\xi = (\xi_j)_{j\in A} \in \widehat{A} = \prod_{j \in I} \widehat{A_j}$ and let $F \subset I$ be a finite set such that $x_j=0$ for $j \notin F$. We claim that, for every $\Delta \supset F$ finite, we have 
\begin{equation*}
    w(x,\xi) \tau_{\Delta} \left( \otimes_{j \in \Delta} f_j\right) = \tau_{\Delta} \left( \otimes_{j \in \Delta} w(x_j,\xi_j) f_j \right), \quad f_j \in \ell^2(A_j).
\end{equation*}
Indeed, one can easily verify that $w(x_j,\xi_j) \delta_j=\delta_j$ for $j \notin F$ (and, in particular, for $j \notin \Delta$) and therefore
\begin{align*}
    w(x,\xi) \tau_{\Delta} \left( \otimes_{j \in \Delta} f_j\right)(y) &= \overline{\xi(2^{-1}x)} \xi(y) \prod_{j \in \Delta} f_j(y_j-x_j) \prod_{j \in I \setminus \Delta} \delta(y_j-x_j) \\
    &= \prod_{j \in \Delta} w(x_j,\xi_j) f_j(y_j) \prod_{j \in I \setminus \Delta} \delta_j(y_j)\\
    &=\tau_{\Delta}\left( \otimes_{j \in \Delta} w(x_j,\xi_j)f_j \right)(y).
\end{align*}
Therefore, for every family $(f_j)_{j \in I}$, $f_j \in \ell^2(A_j)$, for which the limit of $\tau_{\Delta}(\otimes_{j \in \Delta} f_j)$ as $\Delta \to I$ exists, we have
\begin{equation}\label{eq:limit of Weil-Heisenberg}
    w(x,\xi) \left[ \lim_{\Delta \to I} \tau_{\Delta}(\otimes_{j \in \Delta} f_j) \right] = \lim_{\Delta \to I} \tau_{\Delta} \left( \otimes_{j \in \Delta} w(x_j,\xi_j)f_j\right).
\end{equation}
This means (cf. \cite[Lemma 4.4.1]{baez_segal_zhou}) that 
\begin{equation*}
    w(x, \xi) = \bigotimes_{j \in I} w(x_j,\xi_j).
\end{equation*}
\begin{remark}\label{remark:lemma di Segal}
 According to \cite[Lemma 4.4.1]{baez_segal_zhou}, the tensor product $\bigotimes_{j \in I} w(x_j,\xi_j)$ is the unique unitary operator $U$ (if it exists) satisfying 
 \[
U \left[ \lim_{\Delta \to I} \tau_{\Delta}(\otimes_{j \in \Delta} f_j) \right] = \lim_{\Delta \to I} \tau_{\Delta} \left( \otimes_{j \in \Delta} w(x_j,\xi_j)f_j\right)
 \]
 for every family $(f_j)_{j \in I}$, $f_j \in \ell^2(A_j)$, for which the limit of $\tau_{\Delta}(\otimes_{j \in \Delta} f_j)$ as $\Delta \to I$ exists and we have just verified that this relation is satisfied with $U=w(x,\xi)$.
Moreover, always in \cite[Lemma 4.4.1]{baez_segal_zhou}, a condition is given for the existence of the infinite tensor product of unitary operators. In our context, this condition is that the product
\begin{equation*}
    \prod_{j \in I} \langle w(x_j,\xi_j) \delta_j, \delta_j \rangle_{\ell^2(A_j)}
\end{equation*}
converges, which is easily verified since all but finitely many terms of the product are equal to one, because $x = (x_j)_{j\in I} \in A=\bigoplus_{j \in I} A_j$.
\end{remark}
It is interesting to point out that the choice of ground states is crucial for this construction. In fact, the following holds.
\begin{proposition}
    Let $A_0$ be a discrete group, with $A_0 \neq \{0\}$, and let $A_j = A_0$, $j \in I$ for some infinite set of indices $I$. Let $x_0 \in A_0$, $x_0 \neq 0$, and $\xi_0 \in \widehat{A_0}$ such that $\xi_0(x_0) \neq 1$. Let $g \in \ell^2(A_0)$ such that $g(x)=1$ if $x=x_0$ and $g(x)=0$ if $x \neq x_0$. Let $\tau_{\Delta}$, for $\Delta \subset I$ finite, denote the maps in the definition of the grounded tensor product $\bigotimes_{j \in I} (\ell^2(A_j), g)$. Then, there is no unitary operator $U$ on $\bigotimes_{j \in I} (\ell^2(A_j), g)$ such that
    \begin{equation*}
        \lim_{\Delta \to I} \tau_{\Delta} \left(\otimes_{j \in \Delta} w(0,\xi_0) f_j\right) = U \left[ \lim_{\Delta \to I} \tau_{\Delta} \left( \otimes_{j \in \Delta} f_j\right)\right]
    \end{equation*}
    holds for every family $(f_j)_{j \in I}$, $f_j \in \ell^2(A_j)$, for which the limit in the right-hand side exists.
\end{proposition}
\begin{proof}
    As already mentioned in Remark \ref{remark:lemma di Segal}, the existence of such an operator $U$ is equivalent to the product
    \begin{equation*}
        \prod_{j \in I} \langle w(0,\xi_0)g,g \rangle_{\ell^2(A_0)}
    \end{equation*}
    being convergent. However, a direct computation shows that the value of each scalar product is constantly equal to $\overline{\xi_0(x_0)} \neq 1$, and therefore the product cannot converge --- otherwise, the terms of the product would have to converge to 1.
\end{proof}

 Theorem \ref{th:direct product} admits a similar interpretation, that is, we can rephrase the above discussion in the case where $(A_j)_{j \in I}$ is a nonempty family of compact Hausdorff 2-regular Abelian groups, with normalized Haar measure, and the direct sum is replaced by the direct product, that is, \[
A=\prod_{j
\in I}A_j.
\]
In this case, as ground state for every Hilbert space $L^2(A_j)$ we choose the trivial character $g_j \in \widehat{A_j} \subset L^2(A_j)$ (that is, $g_j(x)=1$ for every $x \in A_j$). Then, considering the grounded Hilbert spaces $(L^2(A_j),g_j)$, their grounded tensor product is given by
\begin{equation*}
    (L^2(A),g) \simeq \bigotimes_{j \in I} (L^2(A_j),g_j),
\end{equation*}
where $g \in \widehat{A}$ is the trivial character on $A$ (that is, $g(x)=1$ for every $x \in A$). More precisely, for any $\Delta \subset I$ finite the isometric map $\tau_{\Delta} \colon \bigotimes_{j \in \Delta} L^2(A_j) \to L^2(A)$ is defined on the elements of the kind $\otimes_{j \in \Delta} f_j$, where $f_j \in L^2(A_j)$, as the function of $L^2(A)$ given by
\begin{equation*}
    A \ni x \mapsto \prod_{j \in \Delta} f_j(x_j), \quad x=(x_j)_{j\in A} \in A.
\end{equation*}
Now consider $x = (x_j)_{j\in A} \in A$ and $\xi=(\xi_j)_{j\in A} \in \widehat{A} = \bigoplus_{j \in I} \widehat{A_j}$ and let $F \subset I$ be a  finite set of indices such that $\xi_j = 0$ for $j\in I\setminus F$. Following the same argument as above and using the fact that $w(x_j,\xi_j)g_j = g_j$ for $j \notin F$, for every finite $\Delta \supset F$ we have
\begin{equation*}
    \tau_{\Delta}\left( \otimes_{j \in \Delta} w(x_j,\xi_j)f_j\right) = w(x,\xi) \left[ \tau_{\Delta}(\otimes_{j \in \Delta} f_j)\right], \quad f_j \in L^2(A_j)
\end{equation*}
and so, taking the limit ad $\Delta \to I$, we have verified that
\begin{equation*}
    w(x,\xi) = \bigotimes_{j \in I} w(x_j,\xi_j).
\end{equation*}
\begin{remark}
    In this section, we assumed that the groups $A_j$ were 2-regular. However, this assumption is needed only to define the Weyl--Heisenberg operators. In fact, one can drop this assumption and repeat the above arguments with phase-space shifts $M_{\xi}T_x$ instead of $w(x,\xi)$.
\end{remark}

\section{Further examples}\label{sec examples}
In this section, we give some further examples to which our main theorems apply.
\subsection{Finite groups}
Theorem \ref{thm mainteo 1} applies to every finite Abelian group of odd order, endowed with the discrete topology. This restriction on the order stems from the fact that the doubling map $x \mapsto 2x$ is  injective, that is, 2-regular (since $A$ is finite), if and only if the order of $A$ is odd. Moreover, since the (open compact) subgroup $\{0\} \subset A$ is mapped into itself, the doubling map preserves the measure, that is, $\lambda_2=1$ in \eqref{eq:lambda_2}. As observed in the Introduction, in the case $A=\mathbb{Z}_d^n$, for $d\geq 3$ odd and $n\geq 1$ integer, Theorem \ref{thm mainteo 1} (combined with Remark \ref{rem rem2}) reduces to \cite[Theorem 2]{gross2006hudson}.

\subsection{\texorpdfstring{$p$}{p}-adic fields}
Another basic yet interesting example is given by the $p$-adic numbers $\Q_p$, for some prime number $p$. These are defined as the completion of $\Q$ with respect to the $p$-adic norm (see \cite[Page 34]{follandbook}) and are concretely described as series
\begin{equation}\label{eq:p-adic number series}
    \sum_{k=-\infty}^{\infty} x_k p^k,
\end{equation}
where $x_k \in \{0,\ldots,p-1\}$ for $k\in\mathbb{Z}$ and $x_k = 0$ for $k < m$, for some $m \in \Z$. With the addition defined in the natural way, $\Q_p$ is a locally compact group that is $\sigma$-compact (but not compact), totally disconnected, and also a metric space with the $p$-adic distance. Moreover, being $\Q_p \setminus\{0\}$ an LCA group under multiplication with the same topology, it is immediate to see that the doubling map $x \mapsto 2x$ is an isomorphism, and so $\Q_p$ is 2-regular. The group $\Q_p$ contains the open compact subgroup $\Delta_p$ of the $p$-adic integers, that is, those $p$-adic numbers whose series representation \eqref{eq:p-adic number series} is such that $x_k=0$ for $k < 0$. At this point, the situation divides into the case $p=2$ and $p \geq 3$.

For $p=2$ the doubling map $x \mapsto 2x$ is not a bijection on $\Delta_p$ (because $1/2$ is not a dyadic integer). In fact, since every element of $2\Delta_2$ is represented by a series \eqref{eq:p-adic number series} starting from $m=1$ and since either $x_0=0$ or $x_0=1$, $\Delta_2$ can be decomposed as $\Delta_2 = 2\Delta_2 \cup (1+2\Delta_2)$. Since both $2\Delta_2$ and $1+2\Delta_2$ are compact with the same finite measure, $|\Delta_2|=2|2\Delta_2|$ holds and therefore $\lambda_2=1/2$ in \eqref{eq:lambda_2}. Hence, Theorem \ref{thm mainteo2} applies to $A=\Q_2$.

On the contrary, when $p \geq 3$ the doubling map is a bijection on $\Delta_p$ (because $1/2$ is a $p$-adic integer). Therefore, $\lambda_2=1$ in \eqref{eq:lambda_2} and Theorem \ref{thm mainteo 1} applies to $A=\Q_p$, $p\geq 3$ prime.

\subsection{\texorpdfstring{$n$}{n}-adic groups}
The discussion of the previous section can be extended further to consider the $n$-adic groups $\Omega_n$ (see \cite[Definition 10.2]{hewitt63}), where $n\geq 2$ is any integer. Just like the $p$-adic numbers, the elements of $\Omega_n$ can be represented as a series
\begin{equation*}
    \sum_{k=-\infty}^{\infty} x_k n^k, \quad x_k \in \{0,\ldots,n-1\},
\end{equation*}
where $x_k=0$ for $k < m$, for some $m \in \Z$. Again, this is an LCA group with the addition defined in the natural way and contains the open compact subgroup of $n$-adic integers $\Delta_n$, that is, those elements whose series starts from $m \geq 0$. Since $\Omega_n$ in general is not a field (except when $n=p$ is prime, in which case $\Omega_n=\Q_p$), the proof of the fact that $\Omega_n$ is 2-regular requires some extra work. Before proving the following lemma, we introduce the sets
\begin{equation*}
    U_m = \{x \in \Omega_n \colon x_k = 0\  \mathrm{for\ all\ } k < m\} \quad m \in \Z,
\end{equation*}
which are compact open subgroups of $\Omega_n$ and form a neighborhood basis of $0 \in \Omega_n$ (see \cite[Theorem 10.5]{hewitt63}). In particular, we note that $\Delta_n = U_0$.
\begin{lemma}\label{lem:Omega_n is regular}
    The groups $\Omega_n$ are 2-regular for every $n\geq 2$ integer. Moreover, $\lambda_2=1$ in \eqref{eq:lambda_2} if $n$ is odd and $\lambda_2 < 1$ if $n$ is even.
\end{lemma}
\begin{proof}
    The continuity of the doubling map $x \mapsto 2x$ is clear, so we start by proving that it is injective. Suppose $x \in \Omega_n$ is such that $2x=0$ and $x_k=0$ for $k<m$, which means
    \begin{equation*}
        2\sum_{k=m}^{\infty} x_k n^k = 0.
    \end{equation*}
    If $n$ is odd, the map $x \mapsto 2x$ is an isomorphism of $\Z_n$, so it must be $x_k = 0$ for every $k \geq m$. If $n$ is even, suppose that $x_m \neq 0$. Hence, it must be $2x_m=n$, so $x_m = n/2$, but this implies that $2x_{m+1}+1 \equiv 0$ mod $n$, which is absurd since $n$ is even.

    We can now move to surjectivity. Given $y = \sum_{k = m}^{\infty}y_kn^k$ we want to find $x \in \Omega_n$ such that $2x=y$. We consider two separate cases.
    \begin{itemize}
        \item \textbf{$n$ odd}: we take $x_k = 0$ for $k < m$, while we choose $x_m$ to be the only solution of $2x_m \equiv y_m$ mod $n$. Then, we can rephrase the equation $2x=y$ as
        \begin{equation*}
            an^{m+1} +2\sum_{k = m+1}^{\infty}x_k n^k = \sum_{k=m+1}^{\infty} y_k n^k,
        \end{equation*}
        where $a=0$ if $2x_m < n$ and $a=1$ otherwise and repeat the argument.
        \item \textbf{$n$ even}: we take $x_k = 0$ for $k < m-1$. If $y_m$ is even, we choose $x_{m-1}=0$ and therefore $x_m$ must satisfy $2x_m \equiv y_m$ mod $n$. If $y_{m+1}$ is even, we choose $x_m=y_m/2$, while if $y_{m+1}$ is odd, we choose $x_m=y_m/2+n/2$. On the other hand, if $y_m$ is odd, we choose $x_{m-1}=n/2$ and, therefore, $x_m$ must satisfy $2x_m + 1 \equiv y_m$ mod $n$. If $y_{m+1}$ is even, we choose $x_m = (y_m-1)/2$, while if $y_{m+1}$ is odd, we choose $x_m=(y_m-1)/2+n/2$. 
        
        Then, the equation $2x=y$ becomes 
        \begin{equation*}
            an^{m+1}+2\sum_{k=m+1}^{\infty}x_kn^k = \sum_{k=m+1}^{\infty}y_k n^k,
        \end{equation*}
        where $a=0$ if $y_{m+1}$ is even and $a=1$ otherwise and we can repeat the argument.
    \end{itemize}
    To conclude, we note that in the surjectivity proof we obtained $2U_{m-1} \supset U_m$, which implies that the doubling map $x \mapsto 2x$ is open (since $\{U_m\}_{m \in \Z}$ is a basis of open neighborhoods of 0). Moreover, from the same observation we see that when $n$ is odd, $2U_m=U_m$ and therefore $\lambda_2=1$, while when $n$ is even $2U_m$ is strictly contained in $U_m$ and therefore $\lambda_2 < 1$.
\end{proof}
So, when $n$ is odd, Theorem \ref{thm mainteo 1} applies to $A=\Omega_n$, while when $n$ is even, Theorem \ref{thm mainteo2} applies.

\subsection{Solenoid groups}
To conclude, we show an example where Corollary \ref{cor:A connected compact} applies. To this end, with the notation of the previous section, let $n\geq 2$ be an integer and consider $G = \R \times \Delta_n$. Let $u = 1 \in \Delta_n$ (that is, the series where $x_0=1$ and $x_k=0$ for every $k \geq 1$) and consider the subgroup $H = \{(k,ku)\}_{k \in \Z}$ of $G$. Then, the \emph{$n$-adic solenoid group} is defined as $\Sigma_n \coloneqq G / H$ and is a compact connected Hausdorff Abelian group (see \cite[Definition 10.12, Theorem 10.13]{hewitt63}). The group $\Sigma_n$ can be realized concretely as $[0,1) \times \Delta_n$ with the operation defined as
\begin{equation*}
    (a,x)+(b,y) = (a+b-[a+b],x+y-[a+b]u), \quad (a,x),(b,y) \in [0,1) \times \Delta_n
\end{equation*}
where $[\cdot]$ denotes the floor function (see \cite[Theorem 10.15]{hewitt63}). With this realization, it is easy to see that $\Sigma_n$ is 2-regular if and only if $n$ is even.
\begin{lemma}
    The group $\Sigma_n$ is 2-regular if and only if $n$ is even.
\end{lemma}
\begin{proof}
    We start by proving that $x \mapsto 2x$ is injective if and only if $n$ is even. In fact, suppose that $(a,x) \in \Sigma_n$ is such that $2(a,x)=(0,0)$. Then, either $a=0$ or $a=1/2$. In the former case $x$ must satisfy $2x=0$, but this implies $x=0$ (since $\Omega_n$ is 2-regular), while in the latter case $x$ must solve $2x=u$, but from the proof of Lemma \ref{lem:Omega_n is regular} we know that this equation has a solution in $\Delta_n$ if and only if $n$ is odd. In particular, if $n$ is odd, then $\Sigma_n$ is not 2-regular since $x \mapsto 2x$ is not injective.

    Now, to prove surjectivity when $n$ is even, given $(b,y) \in \Sigma_n$ we want to find $(a,x) \in \Sigma_n$ such that $2(a,x)=(b,y)$, which means
    \begin{equation*}
        2a-[2a]=b, \quad 2x=y+[2a]u \eqqcolon y'.
    \end{equation*}
    From the proof of Lemma \ref{lem:Omega_n is regular} we know that the second equation has a solution in $\Delta_n$ if and only if $y'_0$ is even. So, if $y_0$ is even, we choose $a = b/2$, while if $y_0$ is odd, we choose $a=(b+1)/2$, so that $2x=y+[2a]u$ has a solution in $\Delta_n$.

    To conclude, we notice that for even $n$, the doubling map $x \mapsto 2x$ is a continuous bijective map $\Sigma_n\to \Sigma_n$ between a compact space and a Hausdorff space; hence it is open.
\end{proof}
Therefore, when $n$ is even, Corollary \ref{cor:A connected compact} applies to $\Sigma_n$.
\section*{Acknowledgments}
F.~N. is a fellow of the Accademia delle Scienze di Torino and a member of the Societ\`a Italiana di Scienze e Tecnologie Quantistiche (SISTEQ). 


	

\end{document}